\documentclass[a4,oneside,english,reqno,12pt]{amsart}
\usepackage{extsizes}
\usepackage{geometry}
\usepackage{amssymb,amsmath,amsthm,amsfonts,mathtools,esint,dsfont,bbm,yhmath}
\usepackage{appendix}
\usepackage{enumitem}
\usepackage{datetime}

\usepackage[pdfencoding=auto,psdextra]{hyperref}
\usepackage{bookmark}
\usepackage{cleveref}
\usepackage{crossreftools}
\crefname{subsection}{subsection}{subsections}
\usepackage{orcidlink}
\makeatletter

\usepackage{hyperref} 
\hypersetup{
  pdftitle=Homogeneous attractive Bose--Einstein condensates  with repulsive three-body interactions: the one-dimensional case,
  pdfauthor=Dinh-Thi Nguyen,
  pdfsubject=1D Homogeneous BECs,
}

\theoremstyle{plain}
\newtheorem{theorem}{Theorem}[section]
\newtheorem{lemma}[theorem]{Lemma}
\numberwithin{equation}{section}

\theoremstyle{remark}
\newtheorem{remark}[theorem]{Remark}
\newtheorem*{remark*}{Remark}

\allowdisplaybreaks

\title[1D homogeneous BECs]{Homogeneous attractive Bose--Einstein condensates \\ with repulsive three-body interactions:\\ the one-dimensional case}

\author[D.-T. Nguyen]{Dinh-Thi Nguyen\,\orcidlink{0000-0003-4487-5557}}
\address{Faculty of Mathematics and Computer Science, University of Science, Ho Chi Minh City, Vietnam; and Vietnam National University, Ho Chi Minh City, Vietnam.}
\email{\href{mailto:ndthi@hcmus.edu.vn}{ndthi@hcmus.edu.vn}
}


\begin{document}
\begin{abstract}
We investigate the behavior of a homogeneous one-dimensional Bose gas composed of $N$ identical bosons experiencing attractive two-body and repulsive three-body interactions. Utilizing intermediate Hartree theory, we rigorously derive the homogeneous cubic-quintic nonlinear Schrödinger functional as the mean-field limit of the many-body Hamiltonian. Particular attention is given to the impact of the repulsive three-body term on ground state energy estimates and mass concentration properties
\end{abstract}

\maketitle

\section{Introduction}

In this paper, we investigate the ground-state properties of a one-dimensional (1D) homogeneous Bose gas characterized by competing interactions: an attractive two-body force and a repulsive three-body force. Our primary objective is to establish a mathematically rigorous link between the quantum many-body Hamiltonian and its corresponding effective mean-field models, while analyzing the parameter regimes in which Bose--Einstein condensation (BEC) and energy convergence occur.

We here consider a system of $N\geq 3$ identical bosons in $\mathbb{R}$, described by the nonrelativistic Hamiltonian
\begin{align}\label{hamiltonian}
H_{a,b,N} := & \sum_{i=1}^{N} -\Delta_{x_i} - \frac{a}{N-1}\sum_{1\leq i<j\leq N} U^{\rm 2B}_{N^\alpha}(x_i-x_j) \nonumber \\
& + \frac{b}{(N-1)(N-2)} \sum_{1\leq i<j<k\leq N} U^{\rm 3B}_{N^\beta}(x_i-x_j) U^{\rm 3B}_{N^\beta}(x_i-x_k) ,
\end{align}
acting on the symmetric space $L^{2}_{\rm sym}(\mathbb{R}^{N})$. The parameters $a\geq0$ and $b\geq0$ are the strengths of the two-body and three-body interactions, respectively. The interaction terms $U^{\rm 2B}_{N^\alpha}$ and $U^{\rm 3B}_{N^\beta}$ are scaled through the parameters $\alpha,\beta>0$, i.e.,
\begin{equation}\label{scaled-interaction}
U^{\rm 2B}_{N^\alpha}(x) := N^{\alpha}U^{\rm 2B}(N^\alpha x) \quad \text{and} \quad U^{\rm 3B}_{N^\beta}(x) : = N^{\beta}U^{\rm 3B}(N^\beta x) ,
\end{equation}
in such a way that we may expect a well-defined semiclassical theory in the limit $N\to+\infty$. Here the potential functions $U^{\rm 2B}$ and $U^{\rm 3B}$ satisfy the following common conditions
\begin{equation}\label{condition:2-3-body}
0 \leq U^{\rm kB}(x)=U^{\rm kB}(-x) \quad \text{ with } \quad U^{\rm kB},\widehat{U^{\rm kB}} \in L^1(\mathbb{R}) \quad \text{ and } \quad \int_{\mathbb{R}}U^{\rm kB}(x) {\rm d}x = 1,
\end{equation}
where $k \in \{2,3\}$. In the literature \cite{NamRicTri-22a,NamRicTri-22b,NamRicTri-23,NguRic-24,NguRic-25}, the three-body interaction was represented by a function of the distances between three particles. However, due to technical issues arising in the mathematical studies of homogeneous Bose gases, we consider here the three-body interaction as the product of the same two-body interaction which preserves the properties of three-body interactions.

Our focus lies on the mean-field limit of the quantum ground state energy per particle, given by
\begin{equation}\label{energy:quantum}
E_{a,b,N}^{\mathrm{QM}} := \inf\left\{\left\langle\Psi_N \left|\frac{H_{a,b,N}}{N}\right| \Psi_N \right\rangle : \Psi_N \in L^{2}_{\rm sym}(\mathbb{R}^{N}), \int_{\mathbb R^{N}}|\Psi_N|^{2}=1 \right\}
\end{equation}
and on the associated ground states. In the many-body setting, the mean-field ansatz
\begin{equation}\label{eq:BEC}
\Psi_N(x_1,\ldots,x_N) \approx v^{\otimes N}(x_1,\ldots,x_N) := v(x_1)\ldots v(x_N) \quad \text{with} \quad \|v\|_{L^{2}}^{2}=1
\end{equation}
leads to the homogeneous Hartree functional
\begin{align}\label{functional:h}
\mathcal{E}^{\rm H}_{a, b,N}[v] :={} & \int_{\mathbb{R}}|v'(x)|^{2} {\rm d}x - \frac{a}{2}\iint_{\mathbb{R}^{2}} U^{\rm 2B}_{N^\alpha}(x-y)|v(x)|^{2}|v(y)|^{2} {\rm d}x {\rm d}y \nonumber \\
& + \frac{b}{6}\iiint_{\mathbb{R}^{3}} U^{\rm 3B}_{N^\beta}(x-y)U^{\rm 3B}_{N^\beta}(x-z)|v(x)|^{2}|v(y)|^{2}|v(z)|^{2} {\rm d}x {\rm d}y {\rm d}z .
\end{align}
The associated Hartree energy is given by, with $\lambda>0$,
\begin{equation}\label{energy:h}
E_{a,b,N}^{\mathrm{H}}[\lambda] := \inf\left\{\mathcal{E}^{\rm H}_{a, b,N}[v]:v\in H^1(\mathbb{R}), \int_{\mathbb R}|v|^{2}=\lambda\right\}.
\end{equation}
Note that $E_{a,b,N}^{\mathrm{H}}[1] \geq E_{a,b,N}^{\rm QM}$ holds by testing the variational problem on factorized states \eqref{eq:BEC}. The Hartree functional \eqref{functional:h}, which is commonly used to describe the mean-field approximation in many-body systems, can be formally related to the NLS functional under certain conditions. In the large-$N$ limit, the rescaled potentials converge to delta distributions, and the Hartree functional reduces to the cubic-quintic nonlinear Schr\"odinger (NLS) functional
\begin{equation}\label{functional:nls}
\mathcal{E}^{\rm NLS}_{a,b}[v] := \int_{\mathbb{R}}\left[|v'(x)|^{2} - \frac{a}{2} |v(x)|^{4} + \frac{b}{6} |v(x)|^6\right] {\rm d}x ,
\end{equation}
with the associated NLS ground state energy
\begin{equation}\label{energy:nls}
E^{\rm NLS}_{a,b}[\lambda] := \inf\left\{\mathcal{E}^{\rm NLS}_{a,b}[v]:v\in H^1(\mathbb{R}), \int_{\mathbb{R}}|v(x)|^{2} = \lambda\right\},
\end{equation}
where $\lambda>0$.

The concern of this work is the relation between the three models: the many-body Hamiltonian \eqref{hamiltonian}, the Hartree functional \eqref{functional:h}, and the cubic-quintic NLS functional \eqref{functional:nls}. We establish asymptotic of the ground state energy in the large-$N$ limit and characterize the qualitative behavior of ground states across different parameter regimes. Our results extend previous works on the theory for 1D inhomogeneous Bose gases \cite{Lewin-ICMP,LewNamRou-16,CheHol-17,LewNamRou-17-proc} (see also \cite{BraSacTolHul-95,BraSacHul-97,SacStoHul-98}), as well as studies of Bose systems with three-body interactions \cite{ChePav-11,Chen-12,Yuan-15,Xie-15,CheHol-19,NamSal-20,LiYao-21,NamRicTri-22a,NamRicTri-22b,NamRicTri-23,NguRic-24,NguRic-25} (see also \cite{EsrGreZhoLin-96,JosRic-97,GamFreTom-99,AkhDasVag-99,GamFreTomCho-00}). In particular, the study of 1D homogeneous BECs leads to rich phenomena such as self-trapping \cite{Petrov-14}, supersolid phase \cite{BisBla-15,Blakie-16}, and phase transitions between cubic and cubic-quintic regimes.

In 1D, the cubic nonlinearity is mass-subcritical relative to the kinetic term, making it stable even without the three-body term, while the quintic term provides further stabilization. In the limit of strong repulsive three-body interaction, i.e., $b \to +\infty$, the kinetic term in \eqref{functional:nls} is neglected as three-body repulsion dominates, and the minimization problem \eqref{energy:nls} reduces to a local density approximation. This is the so-called ``Thomas--Fermi limit'' (TF) which appeared in the studies of confined BECs with repulsive two-body interactions \cite{LieSeiYng-00}. Indeed, it was proved in \cite{DoaNgu-26} that, in the Thomas--Fermi limit, we have
\begin{equation}\label{energy:TF}
E^{\rm NLS}_{a,b}[1] \approx E^{\rm TF}_{a,b}[1] = \inf\left\{\mathcal{E}^{\rm TF}_{a,b}[\varrho] : 0\leq \varrho \in L^{1} \cap L^{3}(\mathbb R), \int_{\mathbb R}\varrho(x){\rm d}x = 1\right\} = \frac{a^{2}}{b}E^{\rm TF}_{1,1}[1] < 0
\end{equation}
where
\begin{equation}\label{functional:tf}
\mathcal{E}^{\rm TF}_{a,b}[\varrho] = \int_{\mathbb R} \left[-\frac{a}{2}\varrho(x)^{2} + \frac{b}{6}\varrho(x)^{3}\right]{\rm d}x.
\end{equation}
Furthermore, the minimization problem $E^{\rm TF}_{a,b}[1]$ in \eqref{energy:TF} admits (up to translation) a unique ground state $\varrho^{\rm TF}_{a,b}$, which is positive radially symmetric decreasing and satisfies, by a simple scaling,
\begin{equation}\label{energy:TF-scaling}
\varrho^{\rm TF}_{a,b}(x) = \frac{a}{b}\varrho^{\rm TF}_{1,1}\left(\frac{a}{b}x\right).
\end{equation}
In the following, we summarize the results of \cite{DoaNgu-26} on the existence of homogeneous cubic-quintic NLS ground states as well as its asymptotic behaviors.

\begin{theorem}\label{thm:existence-nls}
We have the following.
\begin{enumerate}[label=(\roman*)]
\item\label{nls-gs-existence} {\bf Existence of NLS ground states.} If $a>0$ and $b \geq 0$ then $E^{\rm NLS}_{a,b}[1] < 0$ and there are ground states, which are nonnegative radially symmetric decreasing.

\item {\bf Existence and uniqueness of TF ground state.} The minimization problem $E^{\rm TF}_{a,b}[1]$ given by \eqref{energy:TF} admits a (unique) ground state for every fixed $a,b>0$. Furthermore, $E^{\rm TF}_{a,b}[1] = -\dfrac{3a^{2}}{8b}$.
\end{enumerate}
\end{theorem}

\begin{theorem}[Mass concentration of NLS ground states]\label{thm:behavior-nls}
Let $\{a_{n}\},\{b_{n}\} \subset (0, +\infty)$ be such that $a_{n} \to a_{0} \in (0,+\infty]$ and $b_{n} \to b_{0} \in [0,+\infty]$. Let $\{v_{n}\}$ be a sequence of ground states of $E^{\rm NLS}_{a_{n},b_{n}}[1]$ given by \eqref{energy:nls}. We have the following.

\begin{enumerate}[label=(\roman*)]
\item {\bf Cubic(-quintic) approximation.} If $0 \leq b_{0} < +\infty$ then, up to a translation and extracting a subsequence,
\begin{equation}\label{cv:gs-cubic-quintic-1d}
\lim_{n\to+\infty} a_{n}^{-\frac{1}{2}}v_{n}(a_{n}^{-1}\cdot) = w_{0}
\end{equation}
strongly in $H^{1}(\mathbb R)$, where $w_{0}$ is a ground state of $E^{\rm NLS}_{1,b_{0}}[1]$. Furthermore,
\begin{equation}\label{cv:energy-cubic-quintic-1d}
\lim_{n\to+\infty} a_{n}^{-2}E^{\rm NLS}_{a_{n},b_{n}}[1] = E^{\rm NLS}_{1,b_{0}}[1].
\end{equation}

\item {\bf Thomas--Fermi approximation.} If $b_{0} = +\infty$ then, up to a translation,
\begin{equation}\label{cv:gs-h-tf}
\lim_{n\to+\infty} \frac{b_{n}}{a_{n}}v_{n}\left(\frac{b_{n}}{a_{n}} \cdot\right)^{2} = \varrho^{\rm TF}_{1,1}
\end{equation}
strongly in $L^{1}\cap L^{3}(\mathbb{R})$ for the whole sequence. Furthermore,
\begin{equation}\label{cv:energy-h-tf}
\lim_{n\to+\infty}\frac{b_{n}}{a_{n}^{2}}E^{\rm NLS}_{a_{n},b_{n}}[1] = E^{\rm TF}_{1,1}[1].
\end{equation}
\end{enumerate}
\end{theorem}

The primary contribution of this article is the characterization of the mass concentration phenomenon within the Hartree framework \eqref{energy:h}. In view of \Cref{thm:existence-nls,thm:behavior-nls}, analogous behavior is expected in the Hartree setting. However, proving the existence of Hartree ground states presents immediate mathematical challenges, as the intermediate potentials lack explicit radial symmetry and monotonicity. We overcome this obstacle by employing Lions' concentration compactness method \cite{Lions-84a}. Furthermore, under suitable structural assumptions on the two-body and three-body interaction profiles, we establish the existence and asymptotic mass concentration of the corresponding ground states. Our main findings are summarized below.

\begin{theorem}[Existence of Hartree ground states]\label{thm:existence-h}
Let $\alpha,\beta>0$ and assume that $U^{\rm 2B}_{N^{\alpha}}$, $U^{\rm 3B}_{N^{\beta}}$ satisfy \eqref{scaled-interaction}-\eqref{condition:2-3-body}. Let $a>0$ and $b \geq 0$ be fixed. Then, the minimization problems \eqref{energy:h} admit ground states for $N$ sufficient large.
\end{theorem}

\begin{theorem}[Mass concentration of Hartree ground states]\label{thm:behaviors-h}
Let $\alpha,\beta>0$ and assume that $U^{\rm 2B}_{N^{\alpha}}$, $U^{\rm 3B}_{N^{\beta}}$ satisfy \eqref{scaled-interaction}-\eqref{condition:2-3-body}. Let $\{a_{N}\} \subset (a_{*},+\infty)$, $\{b_{N}\} \subset (0, +\infty)$ be such that $a_{N} \to a_{0} \in (0,+\infty]$ and $b_{N} \to b_{0} \in [0,+\infty]$, as $N\to+\infty$. Let $\{v_{N}\}$ be a sequence of ground states of $E^{\rm H}_{a_{N},b_{N},N}[1]$ given by \eqref{energy:h}, for $N$ sufficient large. We have the following.
\begin{enumerate}[label=(\roman*)]
\item\label{thm:behavior-h-cubic-quintic} {\bf Cubic-quintic approximation.} Assume that $b_{0} \in [0,+\infty)$, $a_{N} \ll N^{\alpha}$, and $a_{N} \ll N^{\beta}$. Then, up to a translation and extracting a subsequence,
\begin{equation}\label{cv:gs-h-cubic-quintic}
\lim_{N\to+\infty} a_{N}^{-\frac{1}{2}}v_{N}(a_{N}^{-1}\cdot) = w_{0}
\end{equation}
strongly in $H^{1}(\mathbb R)$, where $w_{0}$ is a ground state of $E^{\rm NLS}_{1,b_{0}}[1]$. Furthermore,
\begin{equation}\label{cv:energy-h-cubic-quintic}
\lim_{N\to+\infty} a_{N}^{-2}E^{\rm H}_{a_{N},b_{N},N}[1] = E^{\rm NLS}_{1,b_{0}}[1].
\end{equation}

\item\label{thm:behavior-h-tf} {\bf Thomas--Fermi approximation.} Assume that $b_{0} = +\infty$, $\dfrac{a_{N}^{2}}{b_{N}} \ll N^{2\alpha}$, and $a_{N}b_{N}^{2} \ll N^{\beta}$. Then, up to a translation,
\begin{equation}\label{cv:gs-h-tf}
\lim_{N\to+\infty} \frac{b_{N}}{a_{N}} v_{N}^{*}\left(\frac{b_{N}}{a_{N}}\cdot\right)^{2} = \varrho^{\rm TF}_{1,1}
\end{equation}
strongly in $L^{1}\cap L^{3}(\mathbb{R})$, where $v_{N}^{*}$ is the symmetric decreasing rearrangement of $v_{N}$. Furthermore,
\begin{equation}\label{cv:energy-h-tf}
\lim_{N\to+\infty}\frac{b_{N}}{a_{N}^{2}}E^{\rm H}_{a_{N},b_{N},N}[1] = E^{\rm TF}_{1,1}[1].
\end{equation}
\end{enumerate}
\end{theorem}

\begin{remark}
\begin{itemize}
\item At fixed $a > 0$ and $b \geq 0$, the convergence of $E^{\rm H}_{a,b,N}[1]$ to $E^{\rm NLS}_{a,b}[1]$ as well as of ground states, as $N \to +\infty$, was covered by \eqref{cv:energy-h-cubic-quintic} and \eqref{cv:gs-h-cubic-quintic}. The condition $a_{N} \ll N^{\alpha}$ and $a_{N} \ll N^{\beta}$ are then trivial. 

\item The condition $a_{N}N^{-\beta}  \ll 1$ is to ensure the convergence of the three-body interaction. This condition can be discarded, in the vanishing case of the three-body interaction, i.e., $b_{0} \equiv 0$.

\item In the TF regime, establishing the convergence of Hartree ground states reduces to proving the density convergence of corresponding ``approximate" TF ground states. This convergence is notoriously difficult to establish directly due to the local, non-compact nature of the TF functional. By leveraging the favorable properties of radially symmetric, decreasing functions, we instead establish the strong convergence \eqref{cv:gs-h-tf} for the symmetric decreasing rearrangements of the Hartree ground states. Proving the direct convergence of the original, unrearranged Hartree ground states remains an open problem.
\end{itemize}
\end{remark}

A part of this work is the derivation of the effective homogeneous one-dimensional cubic-quintic nonlinear Schrödinger (NLS) theory as the mean-field model of quantum Bose gases. As reflected in the cubic-quintic NLS framework, the repulsive three-body interaction endows the system with robust stability. Notably, even in the absence of external trapping potentials, the competition between attractive two-body and repulsive three-body forces induces self-trapping behavior. In this setting, the existence of NLS ground states requires an attractive two-body interaction (see \Cref{thm:existence-nls}), while the asymptotic structure of the system is largely driven by the repulsive three-body interaction (see \Cref{thm:behavior-nls}).

Building upon the mean-field behavior observed in the intermediate Hartree theory, we extend our energy estimates to the full $N$-body quantum Hamiltonian \eqref{hamiltonian}. For inhomogeneous Bose gases \cite{LewNamRou-16,Rougerie-20,NguRic-24}, the mean-field limit is typically established via quantum de Finetti theorems \cite{Rougerie-EMS}, which rely on localizing quantum states into finite-dimensional subspaces using spectral projections of the one-body Hamiltonian. This strategy, however, fundamentally requires a confining external potential and is therefore inapplicable to unconfined, homogeneous environments. To circumvent this issue, we employ an approach based on Onsager-type inequalities, as developed for homogeneous systems with two-body interactions only \cite{Lewin-ICMP,LewNamRou-17-proc,DinNguRou-24}. Although this methodology yields lower convergence rates and operates within low-density regimes, it effectively handles unconfined interactions. Recently, this framework was successfully adapted to two-dimensional homogeneous Bose gases with three-body interactions \cite{Nguyen-26-MZ}. Here, we extend these techniques to the one-dimensional setting to obtain the following result.

\begin{theorem}\label{thm:qm}
Under the same assumptions as in \Cref{thm:behaviors-h} and the additional assumptions that $\widehat{U^{\rm 3B}} \geq 0$ and $a_{N}N^{\alpha-1} + b_{N}N^{2\beta-1} \ll \ell_{N}$, where
$$
\ell_{N} = 
\begin{cases}
a_{N}^{2} & \text{if } b_{0} \in [0,+\infty), \\
\dfrac{a_{N}^{2}}{b_{N}} & \text{if } b_{0} = +\infty,
\end{cases}
$$
we have the convergence of the quantum ground state energy
$$
\ell_{N}^{-1}E_{a_{N},b_{N},N}^{\mathrm{QM}} = 
\begin{cases}
E^{\rm NLS}_{1,b_{0}}[1] & \text{if } b_{0} \in [0,+\infty), \\
E^{\rm TF}_{1,1}[1] & \text{if } b_{0} = +\infty.
\end{cases}
$$
\end{theorem}
\medskip

Although \Cref{thm:qm} establishes the convergence of the quantum ground-state energy, it does not assert the condensation of many-body ground states. Indeed, due to translation invariance and the linearity of the $N$-body system, the Hamiltonian \eqref{hamiltonian} fails to admit a true ground state on the unbounded domain. This phenomenon also observed in related unconfined systems (see, e.g., \cite{LieYau-87,DinNguRou-24}). While one may consider sequence-based ``approximate" ground states, BEC is generally not expected because these states naturally form spatial superpositions within the many-body framework. Following the approach suggested in \cite{LieYau-87}, one could instead confine the translation-invariant system to a bounded box matching the scale of the limiting profile, which aligns naturally with the Thomas--Fermi regime \eqref{functional:tf}. In that setting, a Feynman--Hellmann argument could in principle yield mass concentration for suitably localized approximate states. However, this strategy hinges on establishing direct strong convergence for approximate TF states, which remains non-trivial due to the local structure of the TF energy functional. Nevertheless, despite the absence of spatial condensation in the unconfined setting, the quantum ground-state energy per particle reliably converges to the corresponding effective mean-field limit.


\section{Proof of main results}

The asymptotic behavior of the quantum energy $E_{a,b,N}^{\rm QM}$ is obtained by comparing with the NLS energy via the intermediate Hartree energy $E^{\rm H}_{a,b,N}[1]$ defined in \eqref{energy:h}. While the upper bound can be readily obtained by employing trial states $v^{\otimes N}$, the estimation of the energy lower bound presents a technical obstacle. In the case of a trapping potential, this have been done by the quantum de Finetti method (see \cite{NguRic-25}). However, such a method is not applied to our homogeneous model, and we employ the classical Onsager lemma in order to obtain a ``weak'' result. The method of proof comprises three key components:
\begin{enumerate}
\item First, the Hoffmann--Ostenhof inequality \cite{Hof-77} provides a direct estimate of the many-body kinetic energy from below by that of the one-body one.

\item Second, the Onsager’s lemma \cite{Onsager-39} allows for the reduction of the two-body interaction potential of many particles to a one-body term, assuming $U^{\rm 2B}$ is of positive type. To address attractive interactions, it is necessary to combine Onsager’s lemma with a technique developed by Lévy--Leblond \cite{LevLeb-69} (see \cite[Section 3]{Lewin-ICMP} for a comprehensive discussion). The error term in this process is $CaN^{-1}\|U^{\rm 2B}_{N^{\alpha}}\|_{L^{\infty}} = CaN^{\alpha-1}\|U^{\rm 2B}\|_{L^{\infty}}$, where $U^{\rm 2B} \in L^{\infty}(\mathbb R)$ under the assumption $\widehat{U^{\rm 2B}} \in L^{1}(\mathbb R)$.

\item Finally, a modified version of Onsager’s lemma \cite{Nguyen-26-MZ} is specifically tailored for the repulsive three-body interaction of many-particles, enabling its reduction to a modified three-body interaction term involving the square root profile in \eqref{functional:h-modified}. The error term in this process is $CbN^{-1}\|U^{\rm 3B}_{N^{\beta}}\|_{L^{\infty}}^{2} = CbN^{2\beta-1}\|U^{\rm 3B}\|_{L^{\infty}}^{2}$, where $U^{\rm 3B} \in L^{\infty}(\mathbb R)$ since $\widehat{U^{\rm 3B}} \in L^{1}(\mathbb R)$.
\end{enumerate}

In summary, we have the following 1D result analogously to the 2D problem \cite{Nguyen-26-MZ}.

\begin{theorem}\label{thm:qm-h}
Under the assumptions \eqref{scaled-interaction}, \eqref{condition:2-3-body} and assuming further that $\widehat{U^{\rm 3B}} \geq 0$, we have
\begin{equation}\label{cv:energy-qm-hartree}
E^{\rm H}_{a,b,N}[1] \geq E^{\rm QM}_{a,b,N} \geq E^{\rm mH}_{a,b,N} - C\left(aN^{\alpha-1}+bN^{2\beta-1}\right),
\end{equation}
for any $a,b>0$. Here $E^{\rm mH}_{a,b,N}$ is the modified Hartree energy, given by
\begin{equation}\label{energy:h-modified}
E_{a,b,N}^{\rm mH}[1] = \inf\left\{\mathcal{E}_{a,b,N}^{\rm mH}[v] : v \in H^{1}(\mathbb R), \int_{\mathbb R}|v|^{2}=1\right\}
\end{equation}
and the modified Hartree functional $\mathcal{E}_{a,b,N}^{\rm mH}$ is defined as
\begin{align}\label{functional:h-modified}
\mathcal{E}_{a,b,N}^{\rm mH}[v] := & \int_{\mathbb R} |v'(x)|^{2}{\rm d}x - \frac{a}{2}\iint_{\mathbb R^{4}} U^{\rm 2B}_{N^{\alpha}}(x-y)|v(x)|^{2}|v(y)|^{2}{\rm d}x{\rm d}y \nonumber \\
& + \frac{b}{6}\iint_{\mathbb R^{2}}|v(x)|^{2}\sqrt{U^{\rm 3B}_{N^{\beta}}*|v|^{2}}(x)U^{\rm 3B}_{N^{\beta}}(x-y)|v(y)|^{2}\sqrt{U^{\rm 3B}_{N^{\beta}}*|v|^{2}}(y){\rm d}x{\rm d}y.
\end{align}
\end{theorem}

By \Cref{lem:cv-hartree-potentials}, the inequality $E^{\rm H}_{a,b,N}[1] \geq E^{\rm mH}_{a,b,N}[1]$ follows immediately. Advantageously, these two energy levels are asymptotically equivalent, as their deviations from the local nonlinear cubic limit share identical leading-order behavior. Consequently, the existence and asymptotic properties of the modified Hartree ground states mirror those of the original functional presented in \Cref{thm:existence-h,thm:behaviors-h}. Combining \Cref{thm:behaviors-h,thm:qm-h} directly yields Theorem \ref{thm:qm}. Note that the hypothesis a $a_{N}N^{\alpha-1} + b_{N}N^{2\beta-1} \ll \ell_{N}$ ensures that the error terms in \eqref{cv:energy-qm-hartree} remain negligible relative to the leading-order energy scale.  

The remainder of this section is dedicated to proving our core Hartree results (\Cref{thm:existence-h,thm:behaviors-h}). Specifically, we analyze the homogeneous Hartree framework \eqref{energy:h} to establish the existence of ground states and characterize their limiting behavior across various coupling regimes. To this end, we first derive crucial preliminary estimates governing the convergence (and convergence rates) of the two- and three-body interaction potentials, extending previous methodologies developed for two-dimensional homogeneous Bose gases.

\begin{lemma}\label{lem:cv-hartree-potentials}
Let $\alpha,\beta>0$ and assume that $U^{\rm 2B}_{N^{\alpha}}$, $U^{\rm 3B}_{N^{\beta}}$ satisfy \eqref{scaled-interaction}-\eqref{condition:2-3-body}. Then, for every $v\in H^{1}(\mathbb R)$,
\begin{align}
\int_{\mathbb R}|v(x)|^4{\rm d}x & = \lim_{N\to+\infty}\iint_{\mathbb R^{4}} U^{\rm 2B}_{N^{\alpha}}(x-y)|v(x)|^{2}|v(y)|^{2}{\rm d}x{\rm d}y, \label{cv:h-nls-2-body} \\
\int_{\mathbb R}|v(x)|^{6}{\rm d}x & = \lim_{N\to+\infty}\iiint_{\mathbb R^{6}} U^{\rm 3B}_{N^{\beta}}(x-y)U^{\rm 3B}_{N^{\beta}}(x-z)|v(x)|^{2}|v(y)|^{2}|v(z)|^{2}{\rm d}x{\rm d}y{\rm d}z \label{cv:h-nls-3-body} \\
& = \lim_{N\to+\infty}\iint_{\mathbb R^{4}}|v(x)|^{2}\sqrt{U^{\rm 3B}_{N^{\beta}}*|v|^{2}}(x)U^{\rm 3B}_{N^{\beta}}(x-y)|v(y)|^{2}\sqrt{U^{\rm 3B}_{N^{\beta}}*|v|^{2}}(y){\rm d}x{\rm d}y. \label{cv:h-nls-3-body-modified}
\end{align}
Assume further that $xU^{\rm 2B}(x),xU^{\rm 3B}(x) \in L^{1}(\mathbb R)$ then
\begin{align}
0 \leq{}& \int_{\mathbb R}|v(x)|^4{\rm d}x - \iint_{\mathbb R^{4}} U^{\rm 2B}_{N^{\alpha}}(x-y)|v(x)|^{2}|v(y)|^{2}{\rm d}x{\rm d}y \label{cv-rate:h-nls-2-body-0} \\
\leq{}& 2N^{-\alpha} \|v\|_{L^{2}}^{2} \|v'\|_{L^{2}}^{2} \int_{\mathbb R} |xU^{\rm 2B}(x)| {\rm d}x, \label{cv-rate:h-nls-2-body} \\
0 \leq{}& \int_{\mathbb R}|v(x)|^{6} {\rm d}x - \iiint_{\mathbb R^{6}} U^{\rm 3B}_{N^{\beta}}(x-y)U^{\rm 3B}_{N^{\beta}}(x-z) |v(x)|^{2} |v(y)|^{2} |v(z)|^{2} {\rm d}x {\rm d}y {\rm d}z \label{cv-rate:h-nls-3-body-0} \\
\leq{}& 4N^{-\beta} \|v\|_{L^{2}}^{3} \|v'\|_{L^{2}}^{3} \int_{\mathbb R} |xU^{\rm 3B}(x)| {\rm d}x, \label{cv-rate:h-nls-3-body} \\
0 \leq{}& \int_{\mathbb R}|v(x)|^{6} {\rm d}x - \iint_{\mathbb R^{4}}|v(x)|^{2}\sqrt{U^{\rm 3B}_{N^{\beta}}*|v|^{2}}(x)U^{\rm 3B}_{N^{\beta}}(x-y)|v(y)|^{2}\sqrt{U^{\rm 3B}_{N^{\beta}}*|v|^{2}}(y){\rm d}x{\rm d}y \label{cv-rate:h-nls-3-body-modified-0} \\
\leq{}& 5N^{-\beta} \|v\|_{L^{2}}^{3} \|v'\|_{L^{2}}^{3} \int_{\mathbb R} |xU^{\rm 3B}(x)| {\rm d}x. \label{cv-rate:h-nls-3-body-modified}
\end{align}
\end{lemma}
\begin{proof}
The detailed proof of \Cref{lem:cv-hartree-potentials} closely follows \cite[Lemma 2.1]{Nguyen-26-MZ}. The minor discrepancies in the explicit error bounds \eqref{cv-rate:h-nls-2-body-0}, \eqref{cv-rate:h-nls-3-body-0}, and \eqref{cv-rate:h-nls-3-body-modified} stem directly from the one-dimensional Sobolev embedding, i.e.,
\begin{equation}\label{ineq:sobolev}
\|u\|_{L^{\infty}}^{2} \leq \|u\|_{L^{2}}\|u’\|_{L^{2}}.
\end{equation}
\end{proof}

With \Cref{lem:cv-hartree-potentials} in hand, we are now able to prove \Cref{thm:existence-h,thm:behaviors-h} on the existence and behaviors of Hartree ground states.

\begin{proof}[Proof of \Cref{thm:existence-h}]
In contrast to the NLS setting, where ground states exhibit explicit radial symmetry and monotonicity, Hartree ground states lack these structural properties. Consequently, existence is proved via Lions' concentration compactness method \cite{Lions-84a} combined with known variational properties of the NLS framework.

First, we show that the Hartree ground-state energy inherits the negativity of the NLS ground-state energy. Indeed, by the variational principle, \cite[Lemma 7]{LewNamRou-17}, and \eqref{cv-rate:h-nls-2-body-0}, evaluating the functional on the trial state $\ell^{\frac{1}{2}} \phi(\ell\cdot)$ with $\phi \in C_{c}^{\infty}(\mathbb R)$ yields
\begin{align*}
E^{\rm H}_{a,b,N}[1] & \leq \mathcal{E}^{\rm H}_{a,b,N}\left[\ell^{\frac{1}{2}} \phi(\ell\cdot)\right] \\ 
& \leq \ell^{2}\|\phi'\|_{L^{2}}^{2} - \ell a\|\phi\|_{L^{4}}^{4} \left(\frac{1}{2} - \int_{|x|\geq L}U^{\rm 2B}(x){\rm d}x\right) + \ell^{2}aN^{-\alpha} L \|\phi\|_{L^{2}}^{2} \|\phi'\|_{L^{2}}^{2} + \ell^{2}\frac{b}{6}\|\phi\|_{L^{6}}^{6}.
\end{align*}
where $\ell,L>0$ are parameters to be determined. 
Fixing $L>0$ sufficiently large such that
$$
\int_{|x|\geq L}U^{\rm 2B}(x){\rm d}x  < \frac{1}{2}.
$$
and choosing $\ell>0$ sufficiently small, we deduce that for any $a>0$, $b \geq 0$, $\alpha,\beta>0$ and $N>0$,
\begin{equation}\label{energy:hartree-negativity}
E^{\rm H}_{a,b,N}[1] < 0
\end{equation}
provided that $\ell>0$ was chosen small enough.

Next, for fixed parameters $a>0$, $b \geq 0$, and $N>0$, let $\{v_{k}\}$ be a minimizing sequence of \eqref{energy:h}, i.e., $\|v_{k}\|_{L^{2}}^{2} = 1$ and
$$
E^{\rm H}_{a,b,N}[1] = \lim_{k\to+\infty}\mathcal{E}_{a,b,N}^{\rm H}[v_{k}].
$$
Since the attractive two-body interaction is mass-subcritical relative to the kinetic energy, the uniform boundedness of $\{v_{k}\} \in H^{1}(\mathbb R)$ follows directly from estimate \eqref{cv-rate:h-nls-2-body-0}, the non-negativity of the three-body interaction, and the Sobolev embedding \eqref{ineq:sobolev}. Up to translations and subsequence extraction, $v_{k} \to v_{0}$ weakly in $H^{1}(\mathbb R)$ and pointwise almost everywhere in $\mathbb R$. To complete the proof, it suffices to demonstrate that this weak limit is strong in $H^{1}(\mathbb R)$.

We first rule out the vanishing case $v_{0} \equiv 0$ using \eqref{energy:hartree-negativity}. If $v_{0} \not\equiv 0$ then $v_{k} \to 0$ in $L^{p}(\mathbb R)$, for all $p \in (2,+\infty)$. Combining this with \eqref{cv-rate:h-nls-2-body-0} and the non-negativity of the three-body interaction yields
$$
E^{\rm H}_{a,b,N}[1] = \lim_{k\to+\infty}\mathcal{E}_{a,b,N}^{\rm H}[v_{k}] \geq 0.
$$
which contradicts \eqref{energy:hartree-negativity}.

Now, assume $\|v_{0}\|_{L^{2}}^{2}=1$. By the Brezis--Lieb lemma \cite{BreLie-83}, $v_{k} \to v_{0}$ strongly in $L^{2}(\mathbb R)$, and by Sobolev embedding, this strong convergence holds in $L^{p}(\mathbb R)$ for all $2 \leq p<+\infty$. Applying Fatou's lemma then yields
$$
E^{\rm H}_{a,b,N}[1] = \lim_{k\to+\infty}\mathcal{E}_{a,b,N}^{\rm H}[v_{k}] \geq \mathcal{E}_{a,b,N}^{\rm H}[v_{0}] \geq E^{\rm H}_{a,b,N}[1]
$$
which proves that $v_{0}$ is a ground state of $E^{\rm H}_{a,b,N}[1]$.

Finally, we rule out dichotomy, i.e., $0 < \lambda := \|v_{0}\|_{L^{2}}^{2} < 1$. Supposing for contradiction that this holds. the uniform $H^{1}(\mathbb R)$ bound on $\{v_{k}\}$ together with \eqref{cv-rate:h-nls-2-body-0}, \eqref{cv-rate:h-nls-3-body}, and the Brezis--Lieb lemma \cite{BreLie-83} allows us to split the energy as follows
\begin{align}\label{energy:h-split-1}
E^{\rm H}_{a,b,N}[1] = \lim_{k\to+\infty}\mathcal{E}_{a,b,N}^{\rm H}[v_{k}] & \geq \lim_{k\to+\infty}\mathcal{E}_{a,b}^{\rm NLS}[v_{k}] - o(1)_{N\to+\infty} \nonumber \\
& = \mathcal{E}_{a,b}^{\rm NLS}[v_{0}] + \lim_{k\to+\infty}\mathcal{E}_{a,b}^{\rm NLS}[v_{k}-v_{0}] - o(1)_{N\to+\infty} \nonumber \\
& \geq E^{\rm NLS}_{a,b}[\lambda] + E^{\rm NLS}_{a,b}[1-\lambda] - o(1)_{N\to+\infty}.
\end{align}
On the other hand, letting $v_{0}$ be an NLS ground state corresponding to $E^{\rm NLS}_{a,b}[1]$ (which exists for all $a>0$, $b \geq 0$), the variational principle along with \eqref{cv-rate:h-nls-2-body} and \eqref{cv-rate:h-nls-3-body-0} gives
\begin{equation}\label{energy:h-split-2}
E^{\rm H}_{a,b,N}[1] \leq \mathcal{E}_{a,b,N}^{\rm H}[v_{0}] \leq \mathcal{E}_{a,b}^{\rm NLS}[v_{0}] + o(1)_{N\to+\infty} = E^{\rm NLS}_{a,b}[1] + o(1)_{N\to+\infty}.
\end{equation}
Combining \eqref{energy:h-split-1} and \eqref{energy:h-split-2}, we obtain
\begin{equation}\label{energy:h-split-N}
o(1)_{N\to+\infty} + E^{\rm NLS}_{a,b}[1] \geq E^{\rm NLS}_{a,b}[\lambda] + E^{\rm NLS}_{a,b}[1-\lambda].
\end{equation}
To show that inequality \eqref{energy:h-split-N} fails for sufficiently large $N$, recall that $E^{\rm NLS}_{a,b}[\lambda] < 0$ for all $\lambda>0$ and for every $a>0$, $b \geq 0$, with corresponding ground state $v_{\lambda}$. Defining the normalized profile $\widetilde{v_{\lambda}} = v_{\lambda}\big(\|v_{\lambda}\|_{L^{2}}^{2}\cdot\big)$ so that $\|\widetilde{v_{\lambda}}\|_{L^{2}}^{2} = 1$, we compute
\begin{align}\label{energy:h-split-lambda}
E^{\rm NLS}_{a,b}[\lambda] = \mathcal{E}^{\rm NLS}_{a,b}[v_{\lambda}] & = \int_{\mathbb{R}}\left[\|v_{\lambda}\|_{L^{2}}^{-2}|\widetilde{v_{\lambda}}'|^{2} - \frac{a}{2} \|v_{\lambda}\|_{L^{2}}^{2} |\widetilde{v_{\lambda}}|^{4} + \frac{b}{6}\|v_{\lambda}\|_{L^{2}}^{2} |\widetilde{v_{\lambda}}|^{6} \right] \nonumber  \\
& = \lambda\mathcal{E}^{\rm NLS}_{a,b}[\widetilde{v_{\lambda}}] + (\lambda^{-1}-\lambda)\|\widetilde{v_{\lambda}}'\|_{L^{2}}^{2} \nonumber \\
& \geq \lambda E^{\rm NLS}_{a,b}[1] + (1-\lambda^{2)}\|v'_{\lambda}\|_{L^{2}}^{2}.
\end{align}
By identical arguments, we have
\begin{equation}\label{energy:h-split-1-lambda}
E^{\rm NLS}_{a,b}[1-\lambda] \geq (1-\lambda)E^{\rm NLS}_{a,b}[1] + (1-(1-\lambda)^{2})\|v'_{1-\lambda}\|_{L^{2}}^{2}.
\end{equation}
Substituting \eqref{energy:h-split-lambda} and \eqref{energy:h-split-1-lambda} into \eqref{energy:h-split-N} then yields
$$
o(1)_{N\to+\infty} \geq (1-\lambda^{2})\|v'_{\lambda}\|_{L^{2}}^{2} + (1-(1-\lambda)^{2})\|v'_{1-\lambda}\|_{L^{2}}^{2},
$$
which is impossible for large $N$, as the right-hand side is strictly positive and independent of $N$. This eliminates dichotomy, completing the proof of existence.
\end{proof}

\begin{proof}[Proof of \Cref{thm:behaviors-h}]
We divide the proof into two parts corresponding to the cubic-quintic regime and the Thomas--Fermi regime.

Firstly, we establish the cubic(-quintic) approximation for the Hartree problem in the limit $b_{N} \to b_{0}$ with $b_{0} \not\equiv +\infty$. We first derive the upper bound in \eqref{cv:energy-h-cubic-quintic} via the variational principle. Let $v_{0} \in H^{1}(\mathbb R)$ be a ground state for $E^{\rm NLS}_{1,b_{0}}[1]$. Evaluating the Hartree functional on the trial state $a_{N}^{\frac{1}{2}}v_{0}\big(a_{N}\cdot\big)$ and applying bounds \eqref{cv-rate:h-nls-2-body} and \eqref{cv-rate:h-nls-3-body-0} yields
\begin{align*}
E^{\rm H}_{a_{N},b_{N},N}[1] & \leq \mathcal{E}^{\rm H}_{a_{N},b_{N},N}\left[a_{N}^{\frac{1}{2}}v_{0}\big(a_{N}\cdot\big)\right] \\
& \leq a_{N}^{2}\left(\mathcal{E}^{\rm NLS}_{1,b_{N}}[v_{0}] + \|v_{0}\|_{L^{4}}^{4}\int_{|x|\geq L}U^{\rm 2B}(x){\rm d}x + a_{N}N^{-\alpha} L \|v_{0}\|_{L^{2}}^{2} \|v'_{0}\|_{L^{2}}^{2}\right).
\end{align*}
Under the hypothesis $a{N}N^{-\alpha} \ll 1$, choosing $L = L_{N} := a_{N}^{-\frac{1}{2}}N^{\frac{\alpha}{2}} \gg 1$ ensures that the error terms are asymptotically vanishing as $N \to +\infty$, establishing the energy upper bound in \eqref{cv:energy-h-cubic-quintic}. To establish the matching lower bound in \eqref{cv:energy-h-cubic-quintic} and the strong state convergence \eqref{cv:gs-h-cubic-quintic}, we define the rescaled sequence
$$
w_{N} := a_{N}^{-\frac{1}{2}} v_{N} \big(a_{N}^{-1} \cdot \big).
$$
Then $\|w_{N}\|_{L^2} = \|v_{N}\|_{L^2} = 1$ and the Hartree energy functional rewrites as
\begin{align}\label{boundedness:h-scaled-1}
a_{N}^{-2}\mathcal{E}^{\rm H}_{a_{N},b_{N},N}[v_{N}]
= \int_{\mathbb R} \left[|w_{N}'|^{2} - \frac{1}{2}\int_{\mathbb R} |w_{N}|^{2}\left(|w_{N}|^{2}*U^{\rm 2B}_{N^{\alpha}a_{N}^{-1}}\right)+ \frac{b_{N}}{6} |w_{N}|^{2}(U^{\rm 3B}_{N^{\beta}a_{N}^{-1}}*|w_{N}|^{2})^{2}\right].
\end{align}
By the energy upper bound \eqref{cv:energy-h-cubic-quintic}, the left-hand side of \eqref{boundedness:h-scaled-1} is bounded from above uniformly in $N$. Since the attractive two-body potential is mass-subcritical relative to the kinetic energy, applying \eqref{cv-rate:h-nls-2-body-0}, Sobolev inequality \eqref{ineq:sobolev}, and the non-negativity of the three-body term proves that $\{w_{N}\}$ is uniformly bounded in $H^{1}(\mathbb R)$. . Consequently, up to translations and subsequence extraction, $w_{N} \to w_{0}$ weakly in $H^{1}(\mathbb R)$ and almost everywhere in $\mathbb R$. We claim that this convergence holds strongly in $H^{1}(\mathbb R)$. First, we exclude the vanishing case $w_{0} \equiv 0$. If $w_{0} \equiv 0$, Sobolev embedding implies $w_{N} \to 0$ strongly in $L^{p}(\mathbb R)$ for all $p \in (2,+\infty)$. Taking $N\to+\infty$ in \eqref{boundedness:h-scaled-1}, while using \eqref{cv-rate:h-nls-2-body-0} (rescaled by $N^{\alpha}a_{N}^{-1}$), the non-negativity of the three-body term, and upper bound \eqref{cv:energy-h-cubic-quintic} yields
$$
E_{1,b_{0}}^{\rm NLS}[1] \geq \lim_{N\to+\infty}a_{N}^{-2}\mathcal{E}^{\rm H}_{a_{N},b_{N},N}[v_{N}] \geq 0.
$$
This contradicts the negativity $E_{1,b_{0}}^{\rm NLS}[1] < 0$ for all $b_{0} \geq 0$ (see \Cref{thm:existence-nls}) hence $w_{0} \not\equiv 0$. 

Next, we rule out dichotomy, $0 < \lambda := \|w_{0}\|_{L^{2}}^{2} < 1$. Supposing this holds, we employ \eqref{cv-rate:h-nls-2-body-0}, \eqref{cv:h-nls-3-body}, the Brezis--Lieb lemma \cite{BreLie-83}, and the energy upper bound \eqref{cv:energy-h-cubic-quintic} to split the asymptotic energy as follows
\begin{align}\label{energy:h-split}
E_{1,b_{0}}^{\rm NLS}[1] \geq \lim_{N\to+\infty}a_{N}^{-2}\mathcal{E}^{\rm H}_{a_{N},b_{N},N}[v_{N}] & \geq \lim_{N\to+\infty}\mathcal{E}_{1,b_{0}}^{\rm NLS}[w_{N}] \nonumber \\
& = \mathcal{E}_{1,b_{0}}^{\rm NLS}[w_{0}] + \lim_{N\to+\infty}\mathcal{E}_{1,b_{0}}^{\rm NLS}[w_{N}-w_{0}] \nonumber \\
& \geq E_{1,b_{0}}^{\rm NLS}[\lambda] + E_{1,b_{0}}^{\rm NLS}[1-\lambda].
\end{align}
Here, the assumption $a_{N}N^{-\beta} \ll 1$ ensures that the error between the three-body Hartree potential and the quintic NLS nonlinearity vanishes in the limit via \eqref{cv-rate:h-nls-3-body-0}. However, inequality \eqref{energy:h-split} directly contradicts the strict binding inequality
$$
E_{1,b_{0}}^{\rm NLS}[1] < E_{1,b_{0}}^{\rm NLS}[\lambda] + E_{1,b_{0}}^{\rm NLS}[1-\lambda]
$$
which is established analogously to the proof of \Cref{thm:existence-h}. Hence, $\|w_{0}\|_{L^{2}}^{2}=1$. Then, by Brezis--Lieb lemma \cite{BreLie-83} and Sobolev embedding, $w_{N} \to w_{0}$ strongly in $L^{p}(\mathbb R)$ for all $p \in [2,+\infty)$. Taking $N\to+\infty$ in \eqref{boundedness:h-scaled-1}, using \eqref{cv-rate:h-nls-2-body-0}, \eqref{cv:h-nls-3-body}, and applying Fatou's lemma yields
$$
E_{1,b_{0}}^{\rm NLS}[1] \geq \lim_{N\to+\infty}a_{N}^{-2}\mathcal{E}^{\rm H}_{a_{N},b_{N},N}[v_{N}] \geq \lim_{N\to+\infty}\mathcal{E}_{1,b_{0}}^{\rm NLS}[w_{N}] \geq \mathcal{E}_{1,b_{0}}^{\rm NLS}[w_{0}] \geq E_{1,b_{0}}^{\rm NLS}[1].
$$
This proves energy convergence \eqref{cv:energy-h-cubic-quintic} and kinetic energy convergence $\|w_{N}'\|_{L^{2}}^{2} \to \|w_{0}'\|_{L^{2}}^{2}$, which implies $w_{N} \to w_{0}$ strongly in $H^{1}(\mathbb R)$, by Br\'ezis--Lieb lemma \cite{BreLie-83}. This completes the proof of \Cref{thm:behaviors-h}\ref{thm:behavior-h-cubic-quintic}.

We now address the TF approximation for the Hartree problem in the limit $b_{N} \to +\infty$. To derive the energy upper bound in \eqref{cv:energy-h-tf}, we construct a mollified trial state. Let $g_{\mu} = (2\pi)^{-1}\mu e^{-\sqrt{\mu} |x|}$ for $\mu>0$, and denot by $\varrho^{\rm TF}_{1,1}$ the (unique) ground state of $E^{\rm TF}_{1,1}[1]$. Set
$$
v_{N} = \sqrt{g_{\mu}*\varrho^{\rm TF}_{N}} \quad \text{with} \quad \varrho^{\rm TF}_{N}(x) = \frac{a_{N}}{b_{N}}\varrho^{\rm TF}_{1,1}\left(\left(\frac{a_{N}}{b_{N}}\right)x\right)
$$
Since $\int_{\mathbb R}g_{\mu} = 1$ and $|g_{\mu}'| = \sqrt{\mu} g_{\mu}$, the kinetic energy satisfies
\begin{equation}\label{cv:energy-h-tf-kinetic}
|v'_{N}| = \frac{|g_{\mu}'*\varrho^{\rm TF}_{N}|}{2\sqrt{g_{\mu}*\varrho^{\rm TF}_{N}}} \leq \frac{\sqrt{\mu}}{2}\sqrt{g_{\mu}*\varrho^{\rm TF}_{N}}.
\end{equation}
By the variational principle, \cite[Lemma 7]{LewNamRou-17}, \eqref{cv-rate:h-nls-2-body}, \eqref{cv-rate:h-nls-3-body-0}, \eqref{cv:energy-h-tf-kinetic}, and Young's inequality, we obtain
\begin{align*}
E^{\rm H}_{a_{N},b_{N},N}[1] \leq{} & \mathcal{E}^{\rm H}_{a_{N},b_{N},N}[v_{N}] \\
\leq{} & \|v'_{N}\|_{L^{2}}^{2} + a_{N}N^{-\alpha}L\|v'_{N}\|_{L^{2}}\|v_{N}\|_{L^{6}}^{3} \\
& - \left(1-2\int_{|z|\geq L}U(z){\rm d}z\right)\frac{a_{N}}{2} \int_{\mathbb R}|v_{N}(x)|^{4}{\rm d}x + \frac{b_{N}}{6} \int_{\mathbb R}|v_{N}(x)|^{6}{\rm d}x \\
\leq{} & \frac{\mu}{4} + \frac{a_{N}}{2}\sqrt{\mu}N^{-\alpha}L\|\varrho^{\rm TF}_{N}\|_{L^{3}}^{\frac{3}{2}} + (1-o(1))\frac{a_{N}}{2}\left(\int_{\mathbb R}\varrho^{\rm TF}_{N}(x)^{2} - \int_{\mathbb R}(g_{\mu}*\varrho^{\rm TF}_{N})(x)^{2}\right) \\
& - (1-o(1))\frac{a_{N}}{2} \int_{\mathbb R}\varrho^{\rm TF}_{N}(x)^{2}{\rm d}x + \frac{b_{N}}{6} \int_{\mathbb R}\varrho^{\rm TF}_{N}(x)^{3}{\rm d}x \\
={} & \frac{\mu}{4} + \frac{a_{N}^{2}}{2b_{N}}\sqrt{\mu}N^{-\alpha}L\|\varrho^{\rm TF}_{1,1}\|_{L^{3}}^{\frac{3}{2}} + (1-o(1))\frac{a_{N}^{2}}{2b_{N}}\left(\int_{\mathbb R}\varrho^{\rm TF}_{1,1}(x)^{2} - \int_{\mathbb R}\left(g_{\mu \left(\frac{b_{N}}{a_{N}}\right)^{2}}*\varrho^{\rm TF}_{1,1}\right)(x)^{2}\right) \\
& + (1+o(1))E^{\rm TF}_{a_{N},b_{N}}.
\end{align*}
Now, we optimize over $\mu>0$ the first two terms on the right hand side of the above, and we choose $L$ in such a way that the optimal $\mu$ satisfies
$$
\mu = \mu_{N} = \mathcal{O}\left(\left(\frac{a_{N}^{2}}{b_{N}}N^{-\alpha}L\right)^{2}\right) = o(1)E^{\rm TF}_{a_{N},b_{N}}[1] = o\left(\frac{a_{N}^{2}}{b_{N}}\right) \quad \text{and} \quad \mu \left(\frac{b_{N}}{a_{N}}\right)^{2} \gg 1.
$$
This is equivalent to
$$
\frac{N^{\alpha}}{a_{N}} \ll L \ll \sqrt{b_{N}}\frac{N^{\alpha}}{a_{N}}.
$$
The minimum requirement is that $\dfrac{a_{N}^{2}}{b_{N}} \ll N^{2\alpha}$. Under this assumption, we choose
$$
L = L_{N }= b_{N}^{\frac{1}{2}}\frac{N^{\alpha}}{a_{N}} \left(\min\left\{\left(b_{N}^{\frac{1}{2}}\frac{N^{\alpha}}{a_{N}}\right)^{\frac{1}{2}},b_{N}^{\frac{1}{4}}\right\}\right)^{-1} \gg 1.
$$
By standard mollifier properties (e.g., \cite[Theorem 2.16]{LieLos-01}),
$$
\lim_{N\to+\infty}g_{\mu \left(\frac{b_{N}}{a_{N}}\right)^{2}}*\varrho^{\rm TF}_{1,1} = \varrho^{\rm TF}_{1,1}
$$
strongly in $L^{2}(\mathbb R)$, as $N\to+\infty$. Putting all together, we obtain the desired energy upper bound in \eqref{cv:energy-h-tf}.

We prove the matching energy lower bound in \eqref{cv:energy-h-tf} and the density convergence \eqref{cv:gs-h-tf}. Let $\{v_{N}\}$ be a sequence of ground states of $E^{\rm H}_{a_{N},b_{N},N}[1]$, for $a_{N} > 0$, $b_{N} > 0$, and $N$ large enough, and define the rescaled density
$$
\varrho_{N}(x) = \frac{b_{N}}{a_{N}} v_{N}\left(\frac{b_{N}}{a_{N}}x\right)^{2}.
$$
Using \eqref{cv-rate:h-nls-2-body-0}, the non-negativity of the three-body potential, and energy upper bound \eqref{cv:energy-h-tf}, we obtain an a priori kinetic bound
$$
\frac{a_{N}^{2}}{b_{N}}(E^{\rm TF}_{1,1,}[1] + o(1)) \geq \int_{\mathbb R}|v'_{N}|^{2} - \frac{a_{N}}{2}\int_{\mathbb R}|v_{N}|^{4} \geq \frac{1}{2}\int_{\mathbb R}|v'_{N}|^{2} - \frac{a_{N}^{2}}{8}.
$$
Since $b_{N} \to +\infty$, the above yields a priori bound $\|v_{N}'\|_{L^{2}} \leq a_{N}$. Applying \eqref{cv-rate:h-nls-2-body-0}, \eqref{cv-rate:h-nls-3-body}, and kinetic energy non-negativity yields
\begin{align*}
E^{\rm H}_{a_{N},b_{N},N}[1] ={} & \mathcal{E}^{\rm H}_{a_{N},b_{N},N}[v_{N}] \\
\geq{} & -\frac{a_{N}}{2} \int_{\mathbb R}|v_{N}(x)|^{4}{\rm d}x + \left(1-4\int_{|z|\geq L}U^{\rm 3B}(z){\rm d}z\right)\frac{b_{N}}{6} \int_{\mathbb R}|v_{N}(x)|^{6}{\rm d}x \\
& - \frac{2}{3}b_{N}N^{-\beta}L \|v_{N}\|_{L^{2}}^{3} \|v'_{N}\|_{L^{2}}^{3} \\
={} & \frac{a_{N}^{2}}{b_{N}} \left[-\frac{1}{2}\int_{\mathbb R}\varrho_{N}(x)^{2}{\rm d}x + \frac{1-o(1)}{6}\int_{\mathbb R}\varrho_{N}(x)^{3}{\rm d}x - \frac{2}{3}a_{N}b_{N}^{2}N^{-\beta}L\right].
\end{align*}
Under the hypothesis $a_{N}b_{N}^{2} \ll N^{\beta}$, taking $L = (a_{N}b_{N}^{2}N^{-\beta})^{-\frac{1}{2}}$ controls the error term. Combined with the energy upper bound \eqref{cv:energy-h-tf-2}, we obtain
\begin{equation}\label{cv:energy-TF-lower-bound}
o(1) + E^{\rm TF}_{1,1}[1] \geq -\frac{1}{2}\int_{\mathbb R}\varrho_{N}(x)^{2}{\rm d}x + \frac{1-o(1)}{6}\int_{\mathbb R}\varrho_{N}(x)^{3}{\rm d}x.
\end{equation}
By the H\"older inequality, \eqref{cv:energy-TF-lower-bound} yields that $\{\varrho_{N}\}_{N}$ is uniformly bounded in $L^{1}\cap L^{3}(\mathbb R)$. Because the TF functional is local, compactness of, the compactness of $\{\varrho_{N}\}_{N}$ is not directly guaranteed. However, passing to the symmetric decreasing rearrangements $\varrho_{N}^{*}$preserves all $L^{p}$ norms while providing monotonicity. By Helly's selection principle, $\varrho_{N}^{*} \to \varrho_{0}$ weakly in $L^{1} \cap L^{3}(\mathbb R)$ and pointwise almost everywhere in $\mathbb R$. . Monotonicity further implies strong convergence in $L^{p}(\mathbb R)$ for all $p \in (1,3)$, by the arguments in \cite{DoaNgu-26}. We prove that this convergence also holds strongly in $L^{1}(\mathbb R)$ and $L^{3}(\mathbb R)$ as well.

We first note that the limiting profile $\varrho_{0}$ is non-trivial. Note that $\varrho_{0} \equiv 0$ would imply $E^{\rm TF}_{1,1}[1] \geq 0$ from \eqref{cv:energy-TF-lower-bound}, after replacing $\varrho_{N}$ by $\varrho_{N}^{*}$ and taking the limit $N\to\infty$. This contradicts the negativity of $E^{\rm TF}_{1,1}[1] < 0$. Using weak $L^{3}$ limit, strong $L^{2}$ convergence, and Fatou's lemma on \eqref{cv:energy-TF-lower-bound}, we compute
$$
E^{\rm TF}_{1,1}[1] \geq \mathcal{E}^{\rm TF}_{1,1}[\varrho_{0}] = \mathcal{E}^{\rm TF}_{1,1}[\widetilde{\varrho_{0}}]\int_{\mathbb R}\varrho_{0}(x){\rm d}x \geq E^{\rm TF}_{1,1}[1]\int_{\mathbb R}\varrho_{0}(x){\rm d}x \geq E^{\rm TF}_{1,1}[1],
$$
where $\widetilde{\varrho_{0}}(x) = \varrho_{0}\left(\left(\int_{\mathbb R}\varrho_{0}\right)x\right)$ is normalized in $L^{1}(\mathbb R)$. We have utilized the facts that $0 < \int_{\mathbb R}\varrho_{0} \leq 1$ and that $E^{\rm TF}_{1,1}[1] < 0$. Equality must hold throughout, proving that $\int_{\mathbb R}\varrho_{0} = 1$ and that $\int_{\mathbb R}(\varrho_{N}^{*})^{3} \to \int_{\mathbb R}\varrho_{0}^{3}$. By the Br\'ezis--Lieb lemma \cite{BreLie-83}, $\varrho_{N}^{*} \to \varrho_{0}$ strongly in $L^{1} \cap L^{3}(\mathbb R)$. This also establishes the energy lower bound in \eqref{cv:energy-h-tf} and confirms that $\varrho_{0}$ is a ground state of $E^{\rm TF}_{1,1}[1]$. Finally, uniqueness of the limit profile ensures convergence of the whole sequence, concluding the proof.
\end{proof}


\begin{thebibliography}{10}

\bibitem{AkhDasVag-99}
{\sc N.~Akhmediev, M.~P. Das, and A.~Vagov}, {\em {Bose-Einstein condensation
  of atoms with attractive interaction}}, International Journal of Modern
  Physics B, 13 (1999), pp.~625--631.

\bibitem{BisBla-15}
{\sc R.~Bisset and P.~Blakie}, {\em {Crystallization of a dilute atomic dipolar
  condensate}}, Physical Review A, 92 (2015), p.~061603.

\bibitem{Blakie-16}
{\sc P.~B. Blakie}, {\em {Properties of a dipolar condensate with three-body
  interactions}}, Physical Review A, 93 (2016), p.~033644.

\bibitem{BraSacHul-97}
{\sc C.~C. Bradley, C.~Sackett, and R.~Hulet}, {\em {Bose-Einstein condensation
  of lithium: Observation of limited condensate number}}, Physical Review
  Letters, 78 (1997), p.~985.

\bibitem{BraSacTolHul-95}
{\sc C.~C. Bradley, C.~Sackett, J.~Tollett, and R.~G. Hulet}, {\em {Evidence of
  Bose-Einstein condensation in an atomic gas with attractive interactions}},
  Physical review letters, 75 (1995), p.~1687.

\bibitem{BreLie-83}
{\sc H.~Brezis and E.~H. Lieb}, {\em {A Relation Between Pointwise Convergence
  of Functions and Convergence of Functionals}}, Proceedings of the American
  Mathematical Society, 88 (1983), pp.~486--490.

\bibitem{ChePav-11}
{\sc T.~Chen and N.~Pavlovi{\'c}}, {\em {The quintic NLS as the mean field
  limit of a Boson gas with three-body interactions}}, Journal of Functional
  Analysis, 260 (2011), pp.~959--997.

\bibitem{Chen-12}
{\sc X.~Chen}, {\em {Second order corrections to mean field evolution for
  weakly interacting bosons in the case of three-body interactions}}, Archive
  for Rational Mechanics and Analysis, 203 (2012), pp.~455--497.

\bibitem{CheHol-17}
{\sc X.~Chen and J.~Holmer}, {\em {The rigorous derivation of the 2D cubic
  focusing NLS from quantum many-body evolution}}, International Mathematics
  Research Notices, 2017 (2017), pp.~4173--4216.

\bibitem{CheHol-19}
\leavevmode\vrule height 2pt depth -1.6pt width 23pt, {\em {The derivation of
  the $\mathbb{T}^3$ energy-critical NLS from quantum many-body dynamics}},
  Inventiones Mathematicae, 217 (2019), pp.~433--547.

\bibitem{DinNguRou-24}
{\sc V.~D. Dinh, D.-T. Nguyen, and N.~Rougerie}, {\em {Blowup of
  two-dimensional attractive Bose--Einstein condensates at the critical
  rotational speed}}, Annales de l'Institut Henri Poincar{\'e} C, 41 (2024),
  pp.~1055--1081.

\bibitem{DoaNgu-26}
{\sc T.~A.~T. Doan and D.-T. Nguyen}, {\em {Local density approximation and
  other limit regimes for a homogeneous Bose gas with repulsive three-body
  interactions in low dimensional space}}, to appear on Vietnam Journal of
  Mathematics.

\bibitem{EsrGreZhoLin-96}
{\sc B.~Esry, C.~H. Greene, Y.~Zhou, and C.~Lin}, {\em Role of the scattering
  length in three-boson dynamics and bose-einstein condensation}, Journal of
  Physics B: Atomic, Molecular and Optical Physics, 29 (1996), p.~L51.

\bibitem{GamFreTom-99}
{\sc A.~Gammal, T.~Frederico, and L.~Tomio}, {\em {Trapped Bose-Einstein
  condensed gas with two and three-atom interactions}}, in Proceedings of the
  International Workshop, {C. Bertulani, LF. Canto and M. Hussein}, ed., World
  Scientific, 1999.

\bibitem{GamFreTomCho-00}
{\sc A.~Gammal, T.~Frederico, L.~Tomio, and P.~Chomaz}, {\em {Atomic
  Bose-Einstein condensation with three-body interactions and collective
  excitations}}, Journal of Physics B: Atomic, Molecular and Optical Physics,
  33 (2000), p.~4053.

\bibitem{Hof-77}
{\sc M.~{Hoffmann-Ostenhof} and T.~{Hoffmann-Ostenhof}}, {\em
  {{S}chr{\"o}dinger inequalities and asymptotic behavior of the electron
  density of atoms and molecules}}, Physical Review A, 16 (1977),
  pp.~1782--1785.

\bibitem{JosRic-97}
{\sc C.~Josserand and S.~Rica}, {\em {Coalescence and droplets in the
  subcritical nonlinear Schr{\"o}dinger equation}}, Physical Review Letters, 78
  (1997), p.~1215.

\bibitem{LevLeb-69}
{\sc J.-M. L{\'e}vy-Leblond}, {\em Nonsaturation of gravitational forces},
  Journal of Mathematical Physics, 10 (1969), pp.~806--812.

\bibitem{Lewin-ICMP}
{\sc M.~Lewin}, {\em {Mean-field limit of Bose systems: rigorous results}},
  arXiv preprint arXiv:1510.04407,  (2015).

\bibitem{LewNamRou-16}
{\sc M.~Lewin, P.~T. Nam, and N.~Rougerie}, {\em {The mean-field approximation
  and the non-linear Schr\"odinger functional for trapped {B}ose gases}},
  Transactions of the American Mathematical Society, 368 (2016),
  pp.~6131--6157.

\bibitem{LewNamRou-17-proc}
\leavevmode\vrule height 2pt depth -1.6pt width 23pt, {\em Blow-up profile of
  rotating 2d focusing {B}ose gases}, in Workshop on Macroscopic Limits of
  Quantum Systems, Springer, 2017, pp.~145--170.

\bibitem{LewNamRou-17}
\leavevmode\vrule height 2pt depth -1.6pt width 23pt, {\em A note on 2d
  focusing many-boson systems}, Proceedings of the American Mathematical
  Society, 145 (2017), pp.~2441--2454.

\bibitem{LiYao-21}
{\sc Y.~Li and F.~Yao}, {\em {Derivation of the nonlinear Schr{\"o}dinger
  equation with a general nonlinearity and Gross--Pitaevskii hierarchy in one
  and two dimensions}}, Journal of Mathematical Physics, 62 (2021), p.~021505.

\bibitem{LieLos-01}
{\sc E.~H. Lieb and M.~Loss}, {\em {Analysis}}, vol.~14 of {Graduate Studies in
  Mathematics}, American Mathematical Society, Providence, RI, 2nd~ed., 2001.

\bibitem{LieSeiYng-00}
{\sc E.~H. Lieb, R.~Seiringer, and J.~Yngvason}, {\em {Bosons in a trap: A
  rigorous derivation of the {G}ross--{P}itaevskii energy functional}},
  Physical Review A, 61 (2000), p.~043602.

\bibitem{LieYau-87}
{\sc E.~H. Lieb and H.-T. Yau}, {\em {{The {C}handrasekhar theory of Stellar
  Collapse as the Limit of Quantum Mechanics}}}, Communications in Mathematical
  Physics, 112 (1987), pp.~147--174.

\bibitem{Lions-84a}
{\sc P.-L. Lions}, {\em {The concentration-compactness principle in the
  calculus of variations. {T}he locally compact case, {P}art {I}}}, Annales de
  l'Institut Henri Poincar\'e (C) Analyse Non Lin\'eaire, 1 (1984),
  pp.~109--149.

\bibitem{NamRicTri-22b}
{\sc P.~T. Nam, J.~Ricaud, and A.~Triay}, {\em {Dilute Bose gas with three-body
  interaction: Recent results and open questions}}, Journal of Mathematical
  Physics, 63 (2022), p.~061103.

\bibitem{NamRicTri-22a}
\leavevmode\vrule height 2pt depth -1.6pt width 23pt, {\em {Ground state energy
  of the low density Bose gas with three-body interactions}}, Journal of
  Mathematical Physics, 63 (2022), p.~071903.

\bibitem{NamRicTri-23}
\leavevmode\vrule height 2pt depth -1.6pt width 23pt, {\em {The condensation of
  a trapped dilute Bose gas with three-body interactions}}, Probability and
  Mathematical Physics, 4 (2023), pp.~91--149.

\bibitem{NamSal-20}
{\sc P.~T. Nam and R.~Salzmann}, {\em {Derivation of 3D energy-critical
  nonlinear Schr{\"o}dinger equation and Bogoliubov excitations for Bose
  gases}}, Communications in Mathematical Physics, 375 (2020), pp.~495--571.

\bibitem{Nguyen-26-MZ}
{\sc D.-T. Nguyen}, {\em {Homogeneous attractive Bose--Einstein condensates
  with repulsive three-body interactions: the two-dimensional case}},
  Mathematische Zeitschrift, 313 (2026), p.~85.

\bibitem{NguRic-24}
{\sc D.-T. Nguyen and J.~Ricaud}, {\em On one-dimensional bose gases with
  two-body and (critical) attractive three-body interactions}, SIAM Journal on
  Mathematical Analysis, 56 (2024), pp.~3203--3251.

\bibitem{NguRic-25}
\leavevmode\vrule height 2pt depth -1.6pt width 23pt, {\em {Stabilization
  against collapse of 2D attractive Bose--Einstein condensates with repulsive,
  three-body interactions}}, Letters in Mathematical Physics, 115 (2025),
  pp.~1--46.

\bibitem{Onsager-39}
{\sc L.~Onsager}, {\em {Electrostatic Interaction of Molecules}}, Journal of
  Physical Chemistry, 43 (1939), pp.~189--196.

\bibitem{Petrov-14}
{\sc D.~Petrov}, {\em {Three-body interacting bosons in free space}}, Physical
  Review Letters, 112 (2014), p.~103201.

\bibitem{Rougerie-20}
{\sc N.~Rougerie}, {\em {Non linear Schr{\"o}dinger limit of bosonic ground
  states, again}}, Confluentes Mathematici, 12 (2020), pp.~69--91.

\bibitem{Rougerie-EMS}
\leavevmode\vrule height 2pt depth -1.6pt width 23pt, {\em {Scaling limits of
  bosonic ground states, from many-body to nonlinear Schr\"odinger}}, EMS
  Surveys in Mathematical Sciences, 7 (2020), pp.~253--408.

\bibitem{SacStoHul-98}
{\sc C.~Sackett, H.~Stoof, and R.~Hulet}, {\em {Growth and collapse of a
  Bose--Einstein condensate with attractive interactions}}, Physical Review
  Letters, 80 (1998), p.~2031.

\bibitem{Xie-15}
{\sc Z.~Xie}, {\em {Derivation of a nonlinear Schr{\"o}dinger equation with a
  general power-type nonlinearity in $ d= 1, 2$}}, Differential and Integral
  Equations, 28 (2015), pp.~455--504.

\bibitem{Yuan-15}
{\sc J.~Yuan}, {\em {Derivation of the Quintic NLS from many-body quantum
  dynamics in $\mathbb T^2$}}, Communications on Pure \& Applied Analysis, 14
  (2015), p.~1941.

\end{thebibliography}

\end{document}